 \documentclass[a4paper,11pt]{article}
\usepackage[english]{babel}
 \usepackage{cite}      %
 \usepackage[dvips]{graphicx,epsfig,psfrag}  %
 \usepackage{subfigure} %
 \usepackage{url}       %
 \usepackage{amssymb}   %
 \usepackage{amsmath,amsthm}   %
\usepackage{array}     %
\usepackage{times}
\usepackage{footnote}
\usepackage{xspace}
\usepackage{setspace}
 \usepackage[small,normal,bf,up]{caption2}
 
 \usepackage[latin1]{inputenc}
 \usepackage[T1]{fontenc}
 \usepackage{color} 

 \DeclareMathAlphabet{\mathpzc}{OT1}{pzc}{m}{it}

\def\hksqrt{\mathpalette\DHLhksqrt}
\def\DHLhksqrt#1#2{\setbox0=\hbox{$#1\sqrt{#2\,}$}\dimen0=\ht0
\advance\dimen0-0.2\ht0
\setbox2=\hbox{\vrule height\ht0 depth -\dimen0}%
{\box0\lower0.4pt\box2}}
\newcommand{\BibPath}{/home/milan/Documents/University/Research/BibTeX}

\newcommand{\Rl}{\mathbb{R}}

\newcommand{\Nl}{\mathbb{N}}

\newcommand{\fspace}{\;\;\;\;}

\newcommand{\eq}{\triangleq}

\newcommand{\sigsq}{\sigma^{2}}

\newcommand{\w}{\omega}

\newcommand{\Q}{\mathcal{Q}}

\newcommand{\abs}[1]{\left| #1\right|}
\newcommand{\set}[1]{\{ #1 \}}

\newcommand{\setFrTo}[3]{\set{#1}_{#2}^{#3}}
\newcommand{\setFrTok}[3]{\setFrTo{#1}{k=#2}{#3}}
\newcommand{\Set}[1]{\left\{ #1 \right\}}

\newcommand{\sumfromto}[2]{\sum\nolimits_{#1}^{#2}}
\newcommand{\Sumfromto}[2]{\sum\limits_{#1}^{#2}}

\newcommand{\intfromto}[2]{\int_{#1}^{#2}}

\newcommand{\expo}[1]{\textrm{e}^{#1}}

\newcommand{\aeonpipi}{\textrm{a.e. on }[-\pi,\pi]}
\newcommand{\forallwinpipi}{\, \forall \w \in [-\pi,\pi]}

\newcommand{\tr}{\textrm{tr}}
\newcommand{\diag}{\textrm{diag}}
\newcommand{\Bsp}{\mathcal{B}}

\newcommand{\Gsp}{\mathcal{G}}

\newcommand{\Usp}{\mathcal{U}}

\newcommand{\Nsp}{\mathcal{N}}

\newcommand{\by}{\boldsymbol{y}}
\newcommand{\bz}{\boldsymbol{z}}

\newcommand{\bA}{\boldsymbol{A}}

\newcommand{\bI}{\boldsymbol{I}}

\newcommand{\bK}{\boldsymbol{K}}

\newcommand{\bM}{\boldsymbol{M}}

\newcommand{\bQ}{\boldsymbol{Q}}

\newcommand{\bT}{\boldsymbol{T}}

\newcommand{\bW}{\boldsymbol{W}}

\newcommand{\bzero}{\boldsymbol{ 0 }}

\newtheorem{thm}{\textbf{Theorem}}
\newtheorem{defn}{\textbf{Definition}}

\newtheorem{lem}{\textbf{Lemma}}
\newtheorem{rem}{Remark}

\begin{document}
\title{The Quadratic Gaussian Rate-Distortion Function for Source Uncorrelated Distortions}


\author{Milan S. Derpich, Jan {\O}stergaard, and Graham C. Goodwin\\
\small{The University of Newcastle, NSW, Australia}\\
\footnotesize{milan.derpich@studentmail.newcastle.edu.au, \{jan.ostergaard,graham.goodwin\}@newcastle.edu.au }}

\date{ }

\maketitle
\thispagestyle{empty}
\begin{abstract}
 We characterize the rate-distortion function for zero-mean stationary Gaussian sources under the MSE fidelity criterion and subject to the additional constraint that the distortion is uncorrelated to the input. 
The solution is given by two equations coupled through a single scalar parameter.
This has a structure similar to the well known water-filling solution obtained without the uncorrelated distortion restriction.
Our results fully characterize the unique statistics of the optimal distortion.
We also show that, for all positive distortions, the minimum achievable rate subject to the uncorrelation constraint is strictly larger than that given by the un-constrained rate-distortion function.
This gap increases with the distortion and tends to infinity and zero, respectively, as the distortion tends to zero and infinity.
\end{abstract}

\section{Introduction}
Many 
lossy source coding schemes have the property that the end-to-end reconstruction error is uncorrelated with  the source.
We refer to such schemes as \emph{uncorrelated distortion} (UD) coders.
As an example, consider a typical transform coder, as depicted in Fig.~\ref{fig:transform}.
Here, a random vector $X\in \mathbb{R}^N$ is first transformed by an analysis transform $\bT\in \mathbb{R}^{N\times N}$ to yield $U=\bT X$. 
Then $U$ is quantized, yielding the vector, $\hat{U} =\Q(U)$. 
The input signal is finally approximated by $Y=\tilde{\bT}\hat{U}$, where $\tilde{\bT}\in \mathbb{R}^{N\times N}$ is the synthesis transform, cf.~\cite{gergra01,goyal-01}. 
%
\begin{figure}[thb]
\psfrag{x}{$X$}
\psfrag{y}{$Y$}
\psfrag{u}{\hspace{-0.2mm}$U$}
\psfrag{uh}{\hspace{0.9mm}$\hat{U}$}
\psfrag{Q}{\raisebox{-0.2mm}{\hspace{-0.8mm}$\Q$}}
\psfrag{T}{\raisebox{-0.4mm}{\hspace{0mm}$T$}}
\psfrag{Tt}{\raisebox{-0.5mm}{\hspace{0.2mm}$\tilde{T}$}}
\begin{center}
\includegraphics{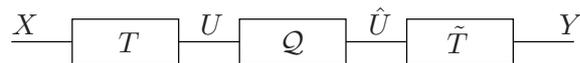}
\caption{Transform coder.}
\label{fig:transform}
\end{center}
\end{figure}
If the quantization error $E \eq \hat{U}-U$ is uncorrelated to $U$, and  if $\bT\tilde{\bT}=\bI$, then it is easy to show that $Y-X$ is uncorrelated to $X$, thus yielding a UD coder. 

More generally, any quantization scheme satisfying the following two properties constitutes a UD coder:
a) The  error introduced by the quantizer is uncorrelated to its input; b) The linear processing (if any) before and after the quantizer yields perfect reconstruction (PR) in the absence of quantization errors. 
%
Property a)
is satisfied in many cases, e.g.\ in high-resolution coding~\cite{marneu05} or when a quantizer with dither (either subtractive~\cite{zamfed96} or non-subtractive~\cite{grasto93})  is employed.
%
On the other hand, the PR condition (Property b)) is often imposed (sometimes implicitly) in the design of filter banks~\cite{vaidya93}, transform coders~\cite{gergra01,goyal-01}, and feedback quantizers~\cite{norsch97,dersil08}.
Thus, any PR source coder using, for example, subtractively dithered quantization, is a UD coder.
The rate-distortion performance of any UD coder can be compared to the underlying Shannon's rate-distortion function $R(D)$ of the source, for a given distortion metric. 
%
%
One may question whether such a comparison is, in fact, fair. 
After all, the additional constraint that the end-to-end distortion is uncorrelated with  the source is not imposed upon $R(D)$. With this in mind, let $R^\perp(D)$ denote the rate distortion function with the additional constraint that the end-to-end distortion is uncorrelated to the source.
(A formal definition of $R^{\perp}(D)$ is given in Section~\ref{sec:UQRDF}). 
Clearly, $R^\perp(D) \!\geq\! R(D)$.%
\footnote{Enforcing additional constraints can never increase the achievable rate region for a given $D$. 
Thus, since the achievable rate region is lower bounded by $R^\perp(D)$ we must have $R^\perp(D)\!\geq\! R(D)$.} 
However, to the best of the authors' knowledge, the problem of characterizing $R^\perp(D)$ has not been formally addressed before.
%
Therefore, questions such as
in which cases (if any) 
$R^\perp(D)$ equals $R(D)$, and
 how $R^\perp(D)$ can be achieved,
appear to be unanswered.


In this paper, we not only give conclusive answers to the above questions, but more importantly, we
completely characterize $R^\perp(D)$ for the quadratic Gaussian case%
\footnote{By quadratic Gaussian we refer to the case of Gaussian sources with the MSE fidelity criterion. 
Moreover, we restrict our attention to zero-mean sources.}  
as a lower bound for the rate achievable under the uncorrelated distortion constraint%
\footnote{A proof of achievability has recently been derived by the authors in~\cite{derost08b}.}.
We show, in Section~\ref{sec:UQRDF}, that $R^\perp(D)$ can be parameterized through a single scalar variable $\alpha\!>\!0$. 
This is a result which parallels the conventional water-filling equations that describe $R(D)$.
%
We characterize the unique optimal statistics that the reconstruction error $Y\!\!-\!\!X$ needs to have in order to achieve $R^\perp(D)$, for a given Gaussian source $X$.
In particular, we show that $Y\!\!-\!\!X$ must be Gaussian.
In addition, we recast the results in a transform coding sense. 
More precisely, we show that if the quantization errors are Gaussian, independent both mutually and from the source, then  the \emph{Karhunen-Lo\`eve Transform} (KLT) is optimal among all perfect reconstruction transforms,  at all rates.%
\footnote{The optimality of a KLT has previously been established by a number of authors in a variety of settings, cf.~\cite{gergra01,huang-63,goyal-00}. However, this appears to be the first time that this result is proven explicitly for $R^\perp(D)$.}
%
%
A comparative analysis between $R^{\perp}(D)$ and $R(D)$ is then presented in Section~\ref{sec:Comparison}.
There we show that $R^\perp(D)$ is convex and monotonically decreasing in $D$, and that 
$R^\perp(D)\!>\!R(D), \forall D\!>\!0$, converging in the limit as $R\rightarrow \infty$.
Furthermore, we show that $R^\perp(D)\!\!\rightarrow\!\! 0\! \Leftrightarrow \! D\rightarrow\! \infty$, which is different from the well known result $R(D)\!=0\! \Leftrightarrow \! D\!\geq\! \sigma_X^2$.%
\footnote{Notice that, in the case of a vanishingly small positive coding rate,  we cannot simply reconstruct the source using its statistical mean (as we would do in a traditional water-filling solution to $R(D)$) since this will lead to linear distortion, clearly correlated to $X$ (since it is a linear function of $X$).}
%

It is worth emphasizing that our results are not tied to any particular source coding architecture,
but are general in the sense that \emph{any} coding scheme 
in which the end-to-end distortion is uncorrelated with the source can do no better than $R^\perp(D)$.

\paragraph*{Notation}
We use uppercase letters to represent random vectors, adding a subscript when referring to one of its elements, i.e., $X_{i}$ is the $i$-th element of the random vector $X$.
The expectation operator is denoted by $\mathbb{E}[\cdot]$.
Uppercase bold letters are used for matrices.
The positive-definite square root of a positive-definite matrix $\bM$ is denoted by $\hksqrt{\bM}$.
We write $\abs{\bM}$ and $\tr(\bM)$ for the determinant and the trace of a matrix $\bM$, respectively.
The \emph{probability density function} (PDF) and covariance matrix of a random (column) vector $X$ are denoted respectively by $f_{X}$ and $\bK_{X}\eq \mathbb{E}[X X^{T}]$, where $X^{T}$ is the transpose of $X$.
We write $\bK_{X,Y}$ for the cross-covariance matrix between two random vectors $X$ and $Y$.
The spectrum of a w.s.s. random process $Z$ with autocorrelation function $R_{Z}[k]\eq \mathbb{E}[Z_{i}Z_{i+k}]$ is denoted by $S_{Z}(\w)\eq\sumfromto{k=-\infty}{\infty}R_{Z}[k]\expo{-jk \w }$, $\forallwinpipi$.
The differential entropy and the differential entropy per dimension of an $N$-length random vector $X$ are denoted, respectively, by $h(X)$ and $\bar{h}(X)\eq \frac{1}{N}h(X)$.
When $X$ is a random process, $\bar{h}(X)\eq \lim_{N\to\infty}\frac{1}{N} h(X_{1},X_{2},\ldots, X_{N})$ denotes the differential entropy rate of $X$.
We use $I(X;Y)$  and $\bar{I}(X;Y)\eq \frac{1}{N}I(X;Y)$ to refer, respectively, to the mutual information and the mutual information per dimension between two random vectors $X$ and $Y$.
When $X$ and $Y$ are random processes , $\bar{I}(X;Y)\eq \lim_{N\to\infty}\frac{1}{N}I(X_{1},\ldots,X_{N} \,;\, Y_{1},\ldots,Y_{N})$ denotes the mutual information rate between $X$ and $Y$.
We write $a.e.$ for ``almost everywhere''.

%



\section{Rate-Distortion Function with Uncorrelated Distortion }\label{sec:UQRDF}
%
We begin by formalizing the definition of the quadratic rate-distortion function under the constraint that the end-to-end distortion be uncorrelated with  the source.
Then, in Section~\ref{subsec:RD_for_Vect},  we characterize this function for Gaussian random vectors,
deferring the case of Gaussian stationary processes to Section~\ref{subsec:RDperp_Rand_Proc}.
\begin{defn}
{The uncorrelated quadratic rate-distortion function $R^{\perp}(D)$ for a random vector (source) $X\in\Rl^{N}$ is defined as}
%
\begin{align}
 	R^{\perp}(D) \eq \min_{
		Y\,:
		\mathbb{E}[X(Y-X)^{T}]= \bzero,\,
		\frac{1}{N}\tr(\bK_{Y-X})\leq D,\,
		\abs{\bK_{X-Y}}>0}
	\bar{I}(X;Y)\label{eq:RperpD_func_def},
\end{align}
where $Y$ is an $N$-length random vector.
%
\end{defn}

\subsection{$R^{\perp}(D)$ for Gaussian Random Vector Sources}\label{subsec:RD_for_Vect}
%
We now present one of the main results of this paper, namely that, for Gaussian vector sources, $R^{\perp}(D)$ is given by two equations linked through a single scalar parameter.
This resembles the ``water-filling'' equations that describe $R(D)$.
The proof of this result, which is presented in Theorem~\ref{thm:RD_Kx}, makes use of the following lemma.
\begin{lem}\label{lem:Z_must_be_Gaussian}
 Let $X\in\Rl^{N}\sim\Nsp(\bzero,\bK_{X})$. Let $Z\in\Rl^{N}$ and $Z_{G}\in\Rl^{N}$ be two random vectors with zero mean and the same covariance matrix, i.e., $\bK_{Z}=\bK_{Z_{G}}$, and having the same cross-covariance matrix with respect to $X$, that is, $\bK_{X,Z}=\bK_{X,Z_{G}}$.
If $Z_{G}$ and $X$ are jointly Gaussian, and if $Z$ has any distribution, then
%
\begin{align}
 I(X;X+Z) \geq I(X;X+Z_{G}).\label{eq:Iineq}
\end{align}
If furthermore  $\abs{\bK_{X+Z}}=\abs{\bK_{X+Z_{G}}}>0$, then equality is achieved in~\eqref{eq:Iineq} iff $Z\sim\Nsp(\bzero,\bK_{Z})$ with $Z$ and $X$ being jointly Gaussian.
\end{lem}

\begin{proof}
Define $Y\eq X+Z$ and $Y_{G} \eq X+Z_{G}$. Then
{\allowdisplaybreaks
\begin{align}
 &I(X;X+Z) - I(X;X+Z_{G}) 
 = h(X|Y_{G}) - h(X|Y) = h(Z_{G}|Y_{G}) - h(Z|Y) \nonumber\\
& \overset{\hphantom{(a)}}{=} 
- \iint f_{Z_{G},Y_{G}}(\bz,\by)\log( f_{Z_{G}|Y_{G}}(\bz|\by))d\bz d\by
    + \iint f_{Z_{},Y_{}}(\bz,\by)\log( f_{Z_{}|Y_{}}(\bz|\by))d\bz d\by \nonumber\\
& \overset{{(a)}}{=} 
- \iint f_{Z_{},Y_{}}(\bz,\by)\log( f_{Z_{G}|Y_{G}}(\bz|\by))d\bz d\by
    + \iint f_{Z_{},Y_{}}(\bz,\by)\log( f_{Z_{}|Y_{}}(\bz|\by))d\bz d\by \nonumber\\
%
%
& \overset{\hphantom{(a)}}{=}
  \int f_{Y}(\by)\int f_{Z_{}|Y_{}}(\bz|\by)\log\left( \frac{f_{Z_{}|Y_{}}(\bz|\by)}{f_{Z_{G}|Y_{G}}(\bz|\by)}\right)d\bz d\by \nonumber \\
& \overset{\hphantom{(a)}}{=}
  \int f_{Y}(\by)D(f_{Z|Y=\by} \Vert f_{Z_{G}|Y_{G}=\by}) d\by 
%
\overset{\hphantom{(a)}}{\geq} 0,\label{eq:diverg_ineq}
\end{align}
}
where $D(f\Vert g)$ is the relative entropy (or \emph{Kullback-Leibler distance}) between the two probability density functions $f$ and $g$.
The equality $(a)$ follows from the fact that $\log(f_{Z_{G}|Y_{G}}(\bz|\by))$ is a quadratic form of $\bz$ and $\by$, and from the fact that $\bK_{Z,Y}=\bK_{Z_{G},Y_{G}}$.
The inequality follows from the fact that $D(f\Vert g)\geq 0$, with equality iff $f=g$.
Thus, equality is achieved iff $f_{Z_{G}|Y_{G}=\by} = f_{Z_{}|Y=\by}$ for all $\by$ such that $f_{Y}(\by)>0$.
The proof is completed by noting that $\abs{\bK_{X+Z}}=\abs{\bK_{X+Z_{G}}}>0$ implies $f_{Y}(\by)>0$ for all $\by\in\Rl^{N}$.
\end{proof}

\begin{rem}
 We note that the above Lemma generalizes Lemma II.2 in~\cite{digcov01}, by relaxing 
the requirement that $f_{Z|X}=f_{Z}$ and  $f_{Z_{G}|X}=f_{Z_{G}}$, to the requirement $\bK_{X,Z}=\bK_{X,Z_{G}}$.
\end{rem}

We are now in a position to present the main result of this section:
%
\begin{thm}\label{thm:RD_Kx}
 %
{Let the source $X\in\Rl^{N}$ be a zero mean Gaussian random vector with positive-definite covariance matrix $\bK_{X}$, having eigenvalues 
$\set{\lambda_{k}}_{k=1}^{N}$.
Then}
\begin{itemize}
 \item[(i)] 
{For any positive $D$,}
%
\begin{align}
 	R^{\perp}(D) 
& = \frac{1}{N}\sumfromto{k=1}{N} \log\left(\frac{\hksqrt{\lambda_{k} +\alpha } + \hksqrt{\lambda_{k}}}{\hksqrt{\alpha}}\right),\label{eq:RperpD_vec} \\
%
\!\!\!\!\!\!\!\!\!\!\!\!\!\!\!\!\!\!\!\!\!\!\!\!\!\!\!\!\!\!\!\!\!\!\!\!\!\!\!\!
\!\!\!\!\!\!\!\!\!\!\!\!\!\!\!\!\!\!\!\!\!\!\!\!\!\!\!\!\!\!\!\!\!\!\!\!\!\!\!\!
\textrm{
where the scalar parameter }&\textrm{$\alpha\in\Rl^{+}$ is such that} \nonumber\\
D 
&= \frac{1}{2N}\sumfromto{k=1}{N} \hksqrt{\lambda_{k}^{2} + \lambda_{k}\alpha  } -  \lambda_{k} \; , \forall D>0.\label{eq:D_of_alpha_vector}
\end{align}
%
\item[(ii)] 
{For each $D>0$, the value of $\alpha$ that satisfies~\eqref{eq:D_of_alpha_vector} is \emph{unique}.}
%
\item[(iii)]
{Let $Y$ satisfy $K_{X(Y-X)}=\bzero$, } $\tr(\bK_{Y-X})/N \leq D$ and $\abs{\bK_{Y-X}}>0$.
{Then}  
$\bar{I}(X;Y) = R^{\perp}(D)$ 
{iff}  $Z\eq (Y-X)\sim\Nsp(\bzero,\bK_{Z^{\star}})$, 
{with}
%
%
\begin{align}\label{eq:Kz_star_def}
 	\bK_{Z^{\star}} \eq  \frac{1}{2}\hksqrt{\bK_{X}^{2} + \alpha \bK_{X} } - \frac{1}{2} \bK_{X},
\end{align}
%
{and where $\alpha$ satisfies~\eqref{eq:D_of_alpha_vector}.}
\end{itemize}
\end{thm}
\begin{proof}
Let $\Usp$ denote the set of all $N$-length random vectors uncorrelated to $X$, and define the sets $\Gsp_{D} \subset \Bsp_{D} \subset \Usp$ as
%
\begin{align}\label{eq:Asp_def}
 \Gsp_{D} \eq \set{ Z \in \Bsp_{D} :  Z \sim \Nsp(\bzero, \bK_{Z})	}; & &
 \Bsp_{D} \eq \set{ Z \in \Usp     :  \tr(\bK_{Z})/N\leq D,\, \abs{\bK_{Z}} >0}.
\end{align}
With the above definitions,~\eqref{eq:RperpD_func_def} can be written as
\begin{align}
 	R^{\perp}(D) & = \min_{Z\in\Bsp_{D}} \Set{	\bar{I}(X;X+Z) 	} 
 \overset{(a)}{=} \min_{Z\in\Gsp_{D}} \Set{	\bar{I}(X;X+Z) 	}
\overset{(b)}{=}	\min_{Z\in\Gsp_{D}}	\Set{	\bar{h}(X+Z) - \bar{h}(Z)			}	\nonumber\\
&\overset{\hphantom{(a)}}{=}	\frac{1}{2N} \!\min_{Z\in\Gsp_{D}}	\!\Set{	\log \abs{\bK_{X} \!+\!\bK_{Z}} \!-\! \log\abs{\bK_{Z}}} 
\overset{\hphantom{(a)}}{=}	\frac{1}{2N}\min_{Z\in\Gsp_{D}}	\!\!\Set{	\log \abs{\bK_{Z}^{-1}\bK_{X}+\bI} }		\label{eq:log_X_I},
\end{align}
where $(a)$ follows directly from Lemma~\ref{lem:Z_must_be_Gaussian} and where $(b)$ holds since the definition of $\Bsp_{D}$ (see~\eqref{eq:Asp_def}), guarantees that both $\bar{h}(X+Z)$ and $\bar{h}(Z)$ exist.
%

We now prove, by contradiction, that the minimizer of $\log|\bK_{Z}^{-1}\bK_{X}+\bI|$ in $\Gsp_{D}$, namely $Z^{\star}$, is such that $\tr(\bK_{Z^{\star}})/N=D$.
For this purpose, suppose that $b \eq ND/\tr(\bK_{Z^{\star}}) >1$, and let $\setFrTok{\zeta_{k}}{1}{N}$ be the eigenvalues of $\bK_{Z^{\star}}^{-1}\bK_{X}$.
Let $Z'\in\Gsp_{D}$ be a Gaussian random vector  with covariance matrix $\bK_{Z'} = b \bK_{Z^{\star}}$.
We then have that $\tr(\bK_{Z'})/N=D$, and that
%
\begin{align}
\log\abs{\bK_{Z^{\star}}^{-1}\bK_{X} + \bI} 
&= \Sumfromto{k=1}{N}\log\left(\zeta_{k} + 1\right)
> \Sumfromto{k=1}{N}\log\left(\frac{\zeta_{k}}{b} + 1\right)
= \log \abs{\bK_{Z'}^{-1}\bK_{X} + \bI },
\end{align}
since $b>1$ and because $\log(\cdot)$ is a strictly increasing function.
Thus $Z^{\star}$ cannot be a minimizer of $\log\abs{\bK_{Z}^{-1}\bK_{X}+\bI}$ in $\Gsp_{D}$ unless $\tr({\bK_{Z^{\star}}})=ND$.

The minimizer of $\log(\abs{\bK_{X}+\bK_{X}}/\abs{\bK_{Z}})$ subject to $\tr(\bK_{Z})/N = D$ can be found using a variational approach.
More precisely, the covariance matrix of the minimizer, $\bK_{Z^{\star}}$,  must \emph{necessarily} be such that the derivative of the Lagrangian
%
\begin{align}
 L(\bK_{Z}) & \eq \log \abs{\bK_{X}+\bK_{Z}} - \log\abs{\bK_{Z}}	+ \beta \tr(\bK_{Z})
\end{align}
with respect to $\bK_{Z}$ 
is zero at $\bK_{Z}=\bK_{Z^{\star}}$, for some $\beta\in\Rl$, 
which
is equivalent to the condition that the matrix differential $\partial L(\bK_{Z}) =0, \, \forall \partial \bK_{Z}$.
Using the fact that $\partial \log\abs{\bM}= \tr(\bM^{-1}\partial \bM)$, for any positive definite matrix $\bM$,
the necessary condition for $Z^{\star}$ to be a minimizer takes the form
\begin{align}
&  \partial L(\bK_{Z})\Big |_{\bK_{Z}=\bK_{Z^{\star}}} 
 \!\!\!=
 \tr\!\left[ (\bK_{X}+\bK_{Z^{\star}}\!)^{-1}  \partial \bK_{Z}\right] - \tr\!\left( \bK_{Z^{\star}}^{-1} \partial \bK_{Z}\right) 
		+ \beta \tr(\partial \bK_{Z} ) 
 = 0,  \forall \partial \bK_{Z} 		\nonumber\\
& \iff 
\tr \left( \left[ (\bK_{X}+\bK_{Z^{\star}})^{-1}   -  \bK_{Z^{\star}}^{-1} 
		+ \beta \bI \right]\partial \bK_{Z}\right)  = 0, \forall \partial \bK_{Z} 		 \nonumber\\
&\iff \!	(\bK_{X} \!+\!\bK_{Z^{\star}} \!)^{-1} \! \! - \! \bK_{Z^{\star}}^{-1} + \beta \bI \! =\! \bzero 		
%
%
%
\iff \!	\bK_{Z^{\star}}\!	 =\!	\pm \frac{1}{2}\hksqrt{\bK_{X}^{2} + \frac{4}{\beta}\bK_{X} } - \frac{1}{2} \bK_{X}
%
.\label{eq:Kz_star}
\end{align}
The fact that $\bK_{Z}^{\star}$ needs to be positive definite implies that $\beta>0$ 
and that it is infeasible to have a negative sign before the square root in~\eqref{eq:Kz_star}.
This, together with the change of variable $\alpha \eq 4/\beta$, leads directly to~\eqref{eq:Kz_star_def}, with $\alpha>0$.
On the other hand, the value of $\alpha$ must be such that the equality constraint $\tr(\bK_{Z^{\star}})/N=D$ is satisfied.
From~\eqref{eq:Kz_star}, and applying Lemma~\ref{lem:Sylvester} (see appendix), this requirement is equivalent to
$
 D = \frac{1}{N}\tr(\bK_{Z^{\star}}) = \frac{1}{2N}\sumfromto{k=1}{N} \hksqrt{\lambda_{k}^{2} + \lambda_{k}\alpha} - \lambda_{k}
$,
which proves~\eqref{eq:D_of_alpha_vector}.
Similarly,~\eqref{eq:RperpD_vec} is obtained by substituting~\eqref{eq:Kz_star} into~\eqref{eq:log_X_I} and then applying Lemma~\ref{lem:Sylvester}, which yields:
\begin{align}
R^{\perp}(D) &=	\frac{1}{2N}\log \abs{\bK_{Z}^{-1}(\bK_{X}+\bK_{Z^{\star}})} 
= \frac{1}{2N}\Sumfromto{k=1}{N} \log\left( \frac{\hksqrt{\lambda_{k}^{2} + \lambda_{k}\alpha} + \lambda_{k} }{\hksqrt{\lambda_{k}^{2} + \lambda_{k}\alpha} - \lambda_{k}}\right)\nonumber\\
&= \frac{1}{2N}\Sumfromto{k=1}{N} \log\left( {\left[\hksqrt{\lambda_{k}^{2} + \lambda_{k}\alpha} + \lambda_{k} \right]^{2}}\Big/{\lambda_{k}\alpha}\right) 
= \frac{1}{N}\Sumfromto{k=1}{N} \log\left[ \frac{\hksqrt{\lambda_{k} + \alpha} + \hksqrt{\lambda_{k}} }{\hksqrt{\alpha}}\right].\nonumber
\end{align}
%
The uniqueness of $\alpha$ is easily verified by noting that the right hand side of~\eqref{eq:D_of_alpha_vector} is monotonically increasing with $\alpha$.
Since $\alpha =1/\beta$ is unique, it follows from~\eqref{eq:Kz_star} that
the covariance matrix of ${Z^{\star}} = \arg \min_{Z\in\Bsp_{D}}\bar{I}(X;X+Z)$  is unique%
\footnote{This is in agreement with the fact that $\log(\abs{\bK_{X} + \bK_{Z}}/\abs{\bK_{Z}})$ is strictly convex in $\bK_{Z}$ for $\abs{\bK_{X}}>0$, as shown in~\cite[Lemma~II.3]{digcov01}, together with the fact that the set $\set{\bK_{Z}: \tr(\bK_{Z})\leq D,\, \abs{\bK_{Z}}>0}$ is convex.},
completing the proof.
\end{proof}
%

\paragraph*{Transform Coding Realization of $R^{\perp}(D)$:}
Closer examination of Lemma~\ref{lem:Sylvester}, when used in~\eqref{eq:Kz_star_def}, suggests that,  for a Gaussian source $X$, $R^{\perp}(D)$ can be achieved by the transform coding architecture shown in Fig.~\ref{fig:transform}.
More precisely, an end-to-end distortion having the optimal covariance matrix $\bK_{Z^{\star}}$ given by~\eqref{eq:Kz_star_def} is obtained by choosing the transform $\bT$ such that $\bT\boldsymbol{\Lambda}\bT^{T}=\bK_{X}$, where $\boldsymbol{\Lambda}\eq\diag(\lambda_{1},\dotsc,\lambda_{N})$ (i.e., the KLT transform for $X$), and by having a Gaussian random vector of quantization errors $E$ with 
$\mathbb{E}[EE^{T}]\!\!=\!\!\frac{1}{2N}\diag(\setFrTok{\hksqrt{\lambda_{k}^{2} + \lambda_{k}\alpha} - \lambda_{k}}{1}{N} )$.%
\footnote{It is easy to show that these noise variances, namely $\sigsq_{E_{k}}$, are such that the derivatives
$
{\partial I(\hat{U}_{k};V_{k})} /{\partial \sigsq_{E_{k} }} 
$
are the same for all $k$.
}
Interestingly, here $\mathbb{E}[E E^{T}]$ is not a scaled identity matrix, as is usually the case in KLT transform coding, cf.~\cite{jaynol84}.
This discrepancy arises from the approximation $\mathbb{E}[E_{k}^{2}] =c \mathbb{E}[U_{k}^{2}] 2^{-2b_{k}}$, commonly used to link the variance of $E_{k}$ to the bit-rate $b_{k}$ at which each $k$-th transform coefficient is quantized.
In this expression, $c\!>\!0$ is a constant that depends on the PDF of the source and on the type of quantizer.
The well known optimal bit allocation analyzed, e.g., in~\cite{jaynol84},
is based upon this formula, and thus minimizes the total bit-rate $r\eq \frac{1}{2N}\sumfromto{k=1}{N}\log_{2}( {\mathbb{E}[U_{k}^{2}]}/{\mathbb{E}[E_{k}^{2}]} ) \!-\! \frac{1}{2}\log_{2}(c)$.
On the other hand,  the optimal quantization errors $E_{k}$ implied by
Theorem~\ref{thm:RD_Kx} need to be Gaussian, their variances being 
such that the end-to-end mutual information
$\bar{I}(X;Y)\!\! = \!\!\frac{1}{2N} \sumfromto{k=1}{N} \!\log_{2}(\mathbb{E}[U_{k}^{2}]  / \mathbb{E}[E_{k}^{2}] + 1)$ is minimized.%
\footnote{The mutual information per dimension between $X$ and $Y$ in this case equals the sum of the mutual informations between each pair $U_{k}$, $E_{k}$, since all these  scalars are mutually independent and $\bT$ is invertible.}
Thus, the difference in the optimal values for $\setFrTok{\mathbb{E}[E_{k}^{2}]}{1}{N}$ obtained in each case is due to the fact that $r\neq \bar{I}(X;Y)$.

%
%
\subsection{$R^{\perp}(D)$ for Gaussian Stationary Random Processes}\label{subsec:RDperp_Rand_Proc}
The $R^{\perp}(D)$ function defined in~\eqref{eq:RperpD_func_def} can be extended to random processes as follows:
%
\begin{defn}
The uncorrelated quadratic rate-distortion function $R^{\perp}(D)$ for a random process $X$ is defined as
\begin{align}\label{eq:Rperp}
 	R^{\perp}(D)  = \min_{  
	Y\,:\,
	\mathbb{E}[X(Y-X)^{T}]=\bzero,\,
	\lim_{N\to\infty}\frac{1}{N}\tr(\bK_{Y-X})\leq D,\,
	\lim_{N\to\infty}\abs{\bK_{Y-X}}^{1/N}>0}
	\bar{I}(X;Y),
\end{align} 
where $Y$ is a random process.
\end{defn}
%
The $R^{\perp}(D)$ function for \emph{stationary} Gaussian random processes can be derived from the results obtained in Section~\ref{subsec:RD_for_Vect}, by restricting to random vectors $X\in\Rl^{N}$ having a Toeplitz covariance matrix, and then letting $N\to\infty$.
More precisely, we have the following result:
%
%
\begin{thm}\label{thm:RD_Sz}
Let the source $X$ be a Gaussian stationary random process with spectrum $S_{X}(\w)$ such that $S_{X}(\w)>0$, $\aeonpipi$.
Then
\begin{itemize}
 \item[(i)] For any $D>0$,
%
\begin{align}
 R^{\perp}(D) 
& 
			= \frac{1}{2\pi}\intfromto{-\pi}{\pi}\log \left(  \frac{\hksqrt{S_{X}(\w) + \alpha} + \hksqrt{S_{X}(\w)} }
				{\hksqrt{\alpha}} \right)d\w,\label{eq:Rperp_def}\\
%
\textrm{where the scalar }& \textrm{parameter $\alpha\in\Rl^{+}$ is such that}\fspace\nonumber\\
%
 D & 
= \frac{1}{4\pi}\intfromto{-\pi}{\pi}
{\left(\hksqrt{S_{X}(\w) +\alpha } - \hksqrt{S_{X}(\w)} \right)\hksqrt{S_{X}(\w)}                 }d\w.\label{eq:D_of_alpha}
\end{align}
%
%
%
\item[(ii)] For each $D>0$, the value of $\alpha$ that satisfies~\eqref{eq:D_of_alpha} is \emph{unique}.
%
\item[(iii)]
Let $Y$ satisfy $\mathbb{E}[X(Y-X)^{T}]=\bzero$,  $\frac{1}{N}\lim_{N\to\infty}\tr(\bK_{Y-X}) \leq D$, $\lim_{N\to\infty} \abs{\bK_{Y-X}}^{\frac{1}{N}} > 0$.
Then
$\bar{I}(X;Y) = R^{\perp}(D)$ \emph{iff} $Z\eq Y-X$ is a Gaussian stationary random process with spectrum 
%
\begin{align}
 S^{\star}_{Z}(\w) & \eq \frac{1}{2}\left(\hksqrt{ S_{X}(\w)  + \alpha } - \hksqrt{S_{X}(\w)}\right)\hksqrt{S_{X}(\w)},\fspace\aeonpipi.\label{eq:Szstar_def_0}
\end{align}
\end{itemize}
\end{thm}

\begin{proof}
Define, from the random processes $X$ and $Y$, the vectors $X^{(N)}\eq [X_{1}\,\cdots \, X_{N}]^{T}$, $Y^{(N)}\eq [Y_{1}\,\cdots \, Y_{N}]^{T}$, $N\in\Nl$.
It is known that $\breve{\lambda}^{(N)} \geq \breve{\lambda}^{(N+1)}, \forall N\in\Nl$,  where $\breve{\lambda}^{(N)}$ and $\breve{\lambda}^{(N+1)}$ are the smallest eigenvalues of $\bK_{X^{(N)}}$ and $\bK_{X^{(N+1)}}$, respectively (see e.g.~\cite[Theorem~4.3.8]{horjoh85}).
This result, together with Lemma~\ref{lem:grensze}, in the Appendix, and the fact that
$S_{X}(\w) >0$, $\aeonpipi$, implies that $\abs{\bK_{X^{(N)}}} >0$, for all $N\in\Nl$.
We can then apply Theorem~\ref{thm:RD_Kx} to each $X^{(N)}$, $\forall N\in\Nl$, obtaining
%
\begin{align}
 R^{\perp(N)}(D) \eq \frac{1}{N}\Sumfromto{k=1}{N} \log
\Bigg(\Bigg( \hksqrt{\lambda_{k}^{(N)} +\alpha^{(N)} } + \hksqrt{\lambda_{k}^{(N)}}\Bigg) \big/{\hksqrt{\alpha^{N}}}
\Bigg),\fspace \forall N\in\Nl,\label{eq:R_N}
\end{align}
where $\alpha^{(N)}$ satisfies
%
\begin{align}
 	D = \frac{1}{2N}\sumfromto{k=1}{N}\hksqrt{\left[\lambda_{k}^{(N)}\right]^{2} + \lambda_{k}^{(N)}\alpha^{(N)}  } -  \lambda_{k}^{(N)},\fspace \forall N\in\Nl, \label{eq:D_N}
\end{align}
and where $\setFrTok{\lambda_{k}^{(N)}}{1}{N}$ denotes the set of  eigenvalues of $\bK_{X^{(N)}}$.
From~\eqref{eq:D_of_alpha}, the optimal distortion covariance matrix for $X^{(N)}$ is
%
\begin{align}
 	\bK_{Z^{\star}}^{(N)} \eq \frac{1}{2}\hksqrt{\bK_{X}^{2} + \alpha \bK_{X}} - \frac{1}{2}\bK_{X}. \label{eq:Kzstr_N}
\end{align}
Direct application of Lemma~\ref{lem:grensze} (see the Appendix) to~\eqref{eq:R_N} and~\eqref{eq:D_N} yields
\begin{align}
 R^{\perp}(D) 
& = \lim_{N\to\infty} R^{\perp(N)}(D)
= \frac{1}{2\pi}\intfromto{-\pi}{\pi}\log \left(  \frac{\hksqrt{S_{X}(\w) + \alpha} + \hksqrt{S_{X}(\w)} }
				{\hksqrt{\alpha}} \right)d\w,
\end{align}
wherein $\alpha\eq \lim_{N\to\infty} \alpha^{(N)}$ is the only scalar that satisfies
%
\begin{align}
 D & 
= \frac{1}{4\pi}\intfromto{-\pi}{\pi}
{\left(\hksqrt{S_{X}(\w) +\alpha } - \hksqrt{S_{X}(\w)} \right)\hksqrt{S_{X}(\w)}                 }.
\end{align}
Similarly, applying Lemma~\ref{lem:grensze} to~\eqref{eq:Kzstr_N} we obtain~\eqref{eq:Szstar_def_0}.
This completes the proof.
\end{proof}
\begin{rem}
 It is interesting to note that the equations characterizing the optimal UD feedback converters derived in~\cite{dersil08} achieve an end-to-end distortion whose spectrum is given precisely by~\eqref{eq:Szstar_def_0}.
Furthermore, it is easy to show that such converters would achieve the $R^{\perp}(D)$ function if the noise due to the scalar quantization within the feedback loop were white Gaussian noise uncorrelated with  the input.
\end{rem}

\section{Comparison with $R(D)$}\label{sec:Comparison}
The next theorem shows that $R^{\perp}(D)$ shares strict monotonicity and convexity with $R(D)$, but deviates from $R(D)$ in the asymptotic limit of large distortions.
%
\begin{thm}\label{thm:convexity}
For any Gaussian random vector (stationary random process) $X$ with positive definite covariance matrix $\bK_{X}$ (with $S_{X}(\w)>0, \aeonpipi$),  the function $R^{\perp}(D)$ is monotonically decreasing and convex.
In addition, 
$D\to\infty \iff R^{\perp}\to 0$, and 
$D\to 0 \iff R^{\perp}\to\infty$.
\end{thm}

\begin{proof}
We present here only the proof for the case of Gaussian random vectors. 
The proof for Gaussian stationary processes proceeds in an analogous fashion.
\emph{Monotonicity:}
We have that 
$\frac{\partial R^{\perp}}{\partial  D} = \frac{\partial R^{\perp}}{\partial\alpha} \big /\frac{\partial D}{\partial\alpha}$,
provided that $\frac{\partial R^{\perp}}{\partial\alpha}$ and $\frac{\partial D}{\partial\alpha}$ exist and that the latter derivative is non-zero.
From~\eqref{eq:RperpD_vec}, we obtain
%
{\allowdisplaybreaks
\begin{align}
	\frac{\partial R^{\perp}}{\partial \alpha } 
&= \frac{1}{N}\sumfromto{k=1}{N} \frac{\partial}{\partial \alpha}\left[ \log\left( \hksqrt{\lambda_{k} + \alpha} + \hksqrt{\lambda_{k}} \right) - \frac{1}{2}\log(\alpha) \right] \nonumber \\
%
%
&=\frac{1}{2N} \sumfromto{k=1}{N} \left[\frac{\hksqrt{\lambda_{k} + \alpha} - \hksqrt{\lambda_{k}} }{\alpha \hksqrt{\lambda_{k} + \alpha} }   - \frac{1}{ \alpha}\right]
= - \frac{1}{2N\alpha} \sumfromto{k=1}{N} \frac{ \hksqrt{\lambda_{k}} }{\hksqrt{\lambda_{k} + \alpha} } . \label{eq:dRperp_d_alpha}
\end{align}
}
On the other hand, from~\eqref{eq:D_of_alpha_vector},
%
\begin{align}
 \frac{\partial D}{\partial \alpha}  \!
& = \!\frac{1}{2N}\!\Sumfromto{k=1}{N} \frac{\partial}{\partial \alpha}\left[ \hksqrt{\lambda_{k}^{2} + \!\lambda_{k}\alpha} - \!\lambda_{k}\right]
\!= \!\frac{1}{4N}\!\Sumfromto{k=1}{N} \left[ \!\frac{\lambda_{k}}{\hksqrt{\lambda_{k}^{2} + \lambda_{k}\alpha}} \!\right] 
\!= \!\frac{1}{4N}\!\Sumfromto{k=1}{N} \left[ \!\frac{ \hksqrt{\lambda_{k}}}{\hksqrt{\lambda_{k} + \alpha}} \!\right],\nonumber
\end{align}
and thus
%
\begin{align}\label{eq:dRdD_v}
 \frac{\partial R^{\perp}(D)}{\partial D} =-\frac{2}{\alpha}, \fspace \forall D>0,
\end{align}
proving that $R^{\perp}(\cdot)$ is a strictly decreasing function (since $\alpha>0$).
\emph{Convexity:} The fact that $\alpha$ grows monotonically with increasing $D$, 
together with~\eqref{eq:dRdD_v}, imply that 
$\frac{\partial R^{\perp}(D)}{\partial D}\big |_{D=D_{1}} > \frac{\partial R^{\perp}(D)} {\partial D}\big |_{D=D_{2}}
\iff
D_{1} > D_{2}
$, and thus $R^{\perp}(\cdot)$ is convex.
\emph{Limits:} It is clear from~\eqref{eq:RperpD_vec} that 
$\lim_{\alpha\to\infty}R^{\perp}=0$ and 
$\lim_{\alpha\to 0}R^{\perp}=\infty$.
Since, as can be seen from~\eqref{eq:dRperp_d_alpha}, $R^{\perp}$ decreases monotonically with increasing $\alpha$, $\forall \alpha\in (0,\infty)$, it follows that
$R^{\perp}\to 0 \iff \alpha \to \infty$, and that $R^{\perp}\to \infty \iff \alpha \to 0$.
Similarly, it follows from~\eqref{eq:D_of_alpha_vector} and the monotonicity of $D$ with respect to $\alpha$ that, for fixed $\setFrTok{\lambda_{k}}{1}{N}$, 
$D\to\infty\iff \alpha \to \infty$ and $D\to 0 \iff \alpha \to 0$.
We then have that
$D\! \to\infty \!\!\iff\! \!R^{\perp}\!\to 0$ and 
$D\to \!0 \!\!\iff\!\! R^{\perp}\!\to\infty$, completing the proof.
\end{proof}

We next show that $R^{\perp}(D)\!>\!R(D)$ for all $D\!>\!0$, converging asymptotically as $D\to 0$.
\begin{thm}\label{thm:gap}
For any Gaussian random vector (stationary random process) $X$ with positive definite covariance matrix $\bK_{X}$ (with $S_{X}(\w)>0, \aeonpipi$), the following holds
 \begin{align}
   \textrm{(i) } R^{\perp}(D) - R(D) & > 0,\fspace \forall D  > 0; & & 
   \textrm{(ii) }\lim_{D\to 0} \left[ R^{\perp}(D) - R(D) \right]  = 0\label{eq:limR_menus_R}.
 \end{align}
\end{thm}

\begin{proof}
We present here only the proof for the case of Gaussian random vectors. 
The proof for Gaussian stationary processes proceeds in an analogous fashion.
Recall that for a Gaussian random vector $X$ having positive covariance matrix with eigenvalues $\setFrTok{\lambda_{k}}{1}{N}$ one has that $R(D)\geq \frac{1}{2N}\sumfromto{k=1}{N}\log(\lambda_{k}/D)$, with equality iff $D\leq \min_{k}\setFrTok{\lambda_{k}}{1}{N}$.
As a consequence,
%
\begin{align}
 R^{\perp}(D) &- R(D) 
%
\overset{\phantom{(a)}}{\leq}
R^{\perp}(D) - \frac{1}{2N}\Sumfromto{k=1}{N}\log(\lambda_{k}/ D) 
\label{eq:substituteme}
\\
&\overset{(a)}{=} 
\frac{1}{2N}\Sumfromto{k=1}{N}\log\left[\frac{\left(\hksqrt{\lambda_{k} + \alpha } + \hksqrt{\lambda_{k}}\,\right)^{2}  }{\alpha}
\left(\frac{1}{2N\lambda_{k}}\Sumfromto{i=1}{N} \hksqrt{\lambda_{i}^{2} + \lambda_{i}\alpha} - \lambda_{i}\right)\right]\nonumber\\
&\overset{\hphantom{(a)}}{=}
\frac{1}{N}
\Sumfromto{k=1}{N}\log\left[\frac{\hksqrt{\lambda_{k} + \alpha } + \hksqrt{\lambda_{k}}  }{\hksqrt{\lambda_{k}}}\right]+
\frac{1}{2}\log\left[\frac{1}{2N}\Sumfromto{i=1}{N} \frac{\lambda_{i}}{\hksqrt{\lambda_{i}^{2} + \lambda_{i}\alpha} + \lambda_{i}}\right].
\label{eq:R_menos_R_util}
%
%
\end{align}
Equality $(a)$ is obtained by substituting~\eqref{eq:RperpD_vec} and~\eqref{eq:D_of_alpha_vector} into the right hand side of~\eqref{eq:substituteme}.
The validity of~\eqref{eq:limR_menus_R}-$(ii)$ then follows directly by taking the limit of the right hand side of~\eqref{eq:R_menos_R_util} as $\alpha\to 0 $.
In order to prove~\eqref{eq:limR_menus_R}-$(i)$, we will show that $D^{\perp}(R) > D(R)$, where the function $D^{\perp}(\cdot)$ is the inverse of $R^{\perp}(\cdot)$ and $D(R)$ is Shannon's distortion-rate function.
For this purpose, consider the random vector $Y' \eq \bW(X+Z^{\star})$, where $\bW\in\Rl^{N\times N}$ is a symmetric, positive-definite matrix, and $Z^{\star}$ is  
as in Theorem~\ref{thm:RD_Kx}-$(iii)$.
Notice that $Y'-X$ and $X$ are \emph{not} uncorrelated unless $\bW=\bI$.
The mutual information per dimension between $Y'$ and $X$ is given by
$
\bar{I}(X;Y')=\bar{h}(\bW Y) - \bar{h}(\bW Y | X ) = \bar{h}(\bW Y) - \bar{h}(\bW X +\bW Z^{\star} | X ) 
=\bar{h}(\bW Y) - \bar{h}(\bW Z^{\star} ) = \bar{h}(Y) - \bar{h}(Z^{\star}) = \bar{I}(X;Y) 
$.
Thus, for any positive definite $\bW$, $\bar{I}(X;Y')=R^{\perp}(D)$.
We next show that for any $D$ (and corresponding $\bK_{Z^{\star}}$), choosing an optimal matrix $\bW$  
yields a $Y'$ whose distortion $\frac{1}{N}\tr(\bK_{Y'-X})$ is strictly smaller than $D$.
It is easy to show that $ \tr(\bK_{Y'-X})$ is minimized by choosing
$
 \bW = \bW^{\star}
$%
, where $\bW^{\star}\eq (\bK_{X} + \bK_{Z^{\star}})^{-1} \bK_{X}$ is the Wiener filter (matrix) for $X+Z^{\star}$.
From this equation, we define 
the function $D'(R)\eq \frac{1}{N} \tr(\bK_{\bW^{\star}(X+Z^{\star})-X})$, describing the distortion associated with
$Y'=\bW^{\star}(X+Z^{\star})$, with the covariance matrix of $Z^{\star}$ as in~\eqref{eq:Kz_star_def} when $R^{\perp}=R$. 
Since $D(R) \leq D'(R)$, we obtain, from applying Lemma~\ref{lem:Sylvester}, and after some algebraic manipulation,  that
%
\begin{align}
 D^{\perp}(R) /  D(R) 
&\geq  D^{\perp}(R) / D'(R) 
= \tr[ \bK_{Z^{\star}}] / \tr[ (\bK_{X}\! +\! \bK_{Z^{\star}}\!)^{-1} \bK_{Z^{\star}}\bK_{X} ] \nonumber \\
& \!\!\!\!\!\!\!\!\!\!\!\!\!\!\!\!\!\!\!\!= \frac{1}{2} 
\left(
\sumfromto{k=1}{N}
{\Big(\hksqrt{\lambda_{k}^{2} + \lambda_{k} \alpha} -  \lambda_{k}\Big)}\alpha 
\right)   
\Big /
\sumfromto{k=1}{N}\Big(\hksqrt{\lambda_{k}^{2} + \lambda_{k} \alpha} -  \lambda_{k} \Big)^{2} 
> 1,\label{eq:srict_ineq}
\end{align}
where the last inequality follows from the fact that $\alpha/2 > \hksqrt{\lambda_{k}^{2} + \lambda_{k} \alpha} -  \lambda_{k}, \, \forall \alpha >0$.
Finally, the fact that both $R(D)$ and $R^{\perp}(D)$ are monotonically decreasing functions, together with~\eqref{eq:srict_ineq},  implies~\eqref{eq:limR_menus_R}-$(i)$.
This completes the proof.
 \end{proof}

The summation term in~\eqref{eq:R_menos_R_util} describes exactly the rate loss $R_{L}(D)\eq R^{\perp}(D) -R(D)$ for all $D\leq \min_{k}\setFrTok{\lambda_{k}}{1}{N}$.
For $D$ within this range, it can be shown, using Chebyshev's sum inequality, that $\partial R_{L}(D)/\partial \alpha >0$, which implies that $R_{L}(D)$ increases monotonically with increasing $D$.
On the other hand, for all $D > 0$, the ratio $D^{\perp}(R)/ D(R)$ can be lower bounded by~\eqref{eq:srict_ineq}.
It can be shown that this bound increases with $\alpha$ (and thus with $D$ as well), tending to $\infty$ as $\alpha\to\infty$, which is in agreement with Theorem~\ref{thm:convexity}.
%
%
%
\section{Concluding Remarks}
In this work we have completely characterized $R^{\perp}(D)$, the quadratic Gaussian rate-distortion function subject to the constraint that the end-to-end distortion be uncorrelated with the source.
We have further proved that this function shares convexity and monotonicity with Shannon's rate-distortion function $R(D)$, but $R^{\perp}(D)$ is positively bounded away from the latter, converging  to $R(D)$ only in the limit as the distortion tends to zero.
We showed that the uncorrelation constraint causes the distortion to unboundedly grow  as the rate tends to zero.
We also discussed the achievability of $R^{\perp}(D)$ for random vectors and stationary random processes through transform coding and feedback quantization architectures.
\section{Appendix}
%
%
%
\begin{lem}[Adapted from Corollary 11.1.2 in~\cite{golvan96}]\label{lem:Sylvester}
 Let 
$$\bA= \bQ \, \diag(\lambda_{1},\ldots,\lambda_{N}) \,\bQ^{-1} = \sumfromto{k=1}{N}\lambda_{k}\bQ_{:,k}\bQ^{-1}_{k,:},$$
with $\bQ\in\mathbb{C}^{N\times N}$, and where $\bQ_{:,k}$ and $\bQ^{-1}_{k,:}$ denote the $k$-th column and the $k$-th row of $\bQ$ and $\bQ^{-1}$, respectively. If $f(\cdot)$ is analytic in a neighbourhood around each $\lambda_{k}$, for $k=1,\ldots,N$, then
%
\begin{align}
 f(A) = \bQ\, \diag(f(\lambda_{1}),\ldots,f(\lambda_{N})) \,\bQ^{-1} = \sumfromto{k=1}{N}f(\lambda_{k})\bQ_{:,k}\bQ^{-1}_{k,:}.
\end{align}
%
\end{lem}
\begin{lem}[
Theorem 4.5.2 in~\cite{berger71}
]\label{lem:grensze}
Let $\bA_{\infty}$ be an infinite Toeplitz matrix with entry $a_{k}\in\mathbb{R}$ on the $k$-th diagonal.
Then the eigenvalues of $\bA_{\infty}$ are contained in the interval $m\leq \lambda \leq M$, where $m$ and $M$ denote the essential infimum and supremum, respectively, of the function
$
 	f(\w) \eq \sumfromto{k=-\infty}{\infty} a_{k}\expo{-jk\w}
$.
Moreover, if both $m$ and $M$ are finite and $G(\lambda)$ is any continuous function of $\lambda\in[m,M]$, then
%
%
\begin{align}
 \lim_{N\to\infty} \frac{1}{N}\sumfromto{k=1}{N} G(\lambda_{k}^{(N)}) = \frac{1}{2\pi}\intfromto{-\pi}{\pi}G[f(\w)]d\w,
\end{align}
where the $\lambda^{(N)}$ are the eigenvalues of the sub-matrix $\bA^{(N)}\in\Rl^{N\times N}$ of $\bA_{\infty}$ centered about the main diagonal of $\bA_{\infty}$.
\end{lem}
%


{
\begin{spacing}{0.2}
	\bibliographystyle{\BibPath/IEEEtran}
\end{spacing}
}

\end{document}